\documentclass[letterpaper, 10 pt, conference]{IEEE/ieeeconf}  

 \IEEEoverridecommandlockouts                              
\usepackage{cite}
\usepackage{amsmath,amssymb,amsfonts}
\usepackage{algorithmic}
\usepackage{textcomp}
\usepackage{xcolor}
\usepackage{multirow}
\usepackage[normalem]{ulem}
\useunder{\uline}{\ul}{}

\usepackage{graphics} 
\usepackage{epsfig} 
\usepackage{mathptmx} 
\usepackage{times} 
\usepackage{amsmath} 
\usepackage{amsfonts}
\usepackage[cal=cm,calscaled=.96]{mathalfa}

\usepackage{amsthm}
\usepackage{amssymb}
\usepackage{acro}
\usepackage{todonotes}
\usepackage{tikz}
\usetikzlibrary{arrows,automata}

\def\BibTeX{{\rm B\kern-.05em{\sc i\kern-.025em b}\kern-.08em
    T\kern-.1667em\lower.7ex\hbox{E}\kern-.125emX}}

\theoremstyle{plain}

\newtheorem{question}{Question}

\newtheorem{lemma}{Lemma}

\theoremstyle{definition}

\newtheorem{assumption}{Assumption}

\newtheorem{definition}{Definition}

\theoremstyle{remark}
\newtheorem*{remark}{Remark}

\newcommand{\fig}[1]{\mbox{Fig.~\ref{#1}}}

\newcommand{\sect}[1]{Section~\ref{#1}}

\newcommand{\defn}[1]{Definition~\ref{#1}}

\def\ap{\ensuremath{{AP}}}

\def\game{\ensuremath{\mathcal{G}}}	
\def\robot{\ensuremath{R}}	
\def\adv{\ensuremath{E}}	

\def\act{\ensuremath{Act}}	
\def\robotact{\ensuremath{Act_\robot}}	
\def\advact{\ensuremath{Act_\adv}}	

\def\tsys{\ensuremath{\mathcal{TS}}}	
\def\aut{\ensuremath{\mathcal{A}}}	
\def\final{\ensuremath{F}}

\def\att{\ensuremath{\mathsf{Attr}}}

\newcommand{\win}{\ensuremath{\mathsf{Win}}}

\def\labelfcn{\ensuremath{L}}

\def\distr{\ensuremath{{\mathcal{D}}}}
\def\indicator{\ensuremath{\mathbf{1}}}

\newcommand{\calF}{\mathcal{F}}

\DeclareAcronym{ltl}{
	short = LTL, long = Linear Temporal Logic ,
	class = abbrev
}

\DeclareAcronym{scltl}{
	short = sc-LTL, long = Syntactically Co-safe LTL ,
	class = abbrev
}

\DeclareAcronym{tsys}{
	short = TS, long = Transition System ,
	class = abbrev
}

\DeclareAcronym{dfa}{
	short = DFA, long = Deterministic Finite Automaton ,
	class = abbrev
}

\DeclareAcronym{mdp}{
	short = MDP, long = Markov Decision Process ,
	class = abbrev
}

\newcommand{\stopgame}{\mathsf{stop}}
\newcommand{\sink}{\mathsf{sink}}

\title{\LARGE \bf
Opportunistic Synthesis in Reactive Games under Information Asymmetry
}

\author{Abhishek N. Kulkarni and Jie Fu
\thanks{This material is based upon work supported by the Defense Advanced Research Projects Agency (DARPA) under Agreement No. HR00111990015. }
\thanks{Abhishek N. Kulkarni is  with Robotics Engineering Program, Worcester Polytechnic Institute, Worcester, MA 01604, USA
        {\tt\small  ankulkarni@wpi.edu}}%
\thanks{Jie Fu is with the faculty of the Department of Electrical and Computer Engineering, with affiliation to Robotics Engineering Program, Worcester Polytechnic Institute, Worcester, MA 01604, USA
        {\tt\small jfu2@wpi.edu}}%
}

\begin{document}

\maketitle
\thispagestyle{empty}
\pagestyle{empty}

\begin{abstract}

Reactive synthesis is a class of methods to construct a provably-correct control system, referred to as a robot, with respect to a temporal logic specification in the presence of a dynamic and uncontrollable environment. This is achieved by modeling the interaction between the robot and its environment as a two-player zero-sum game. However, existing reactive synthesis methods assume both players to have complete information, which is not the case in many strategic interactions. In this paper, we use a variant of hypergames to model the interaction between the robot and its environment; which has incomplete information about the specification of the robot. This model allows us to identify a subset of game states from where the robot can leverage the asymmetrical information to achieve a better outcome, which is not possible if both players have symmetrical and complete information. We then introduce a novel method of  \textit{opportunistic synthesis} by defining a \ac{mdp} using the hypergame under temporal logic specifications. When  the environment plays some stochastic strategy in its perceived sure-winning and sure-losing regions of the game,  we show that by following the opportunistic strategy, the robot is ensured to only improve the outcome of the game---measured by satisfaction of sub-specifications---whenever an opportunity becomes available. We demonstrate the correctness and optimality of this method using a robot motion planning example in the presence of an adversary.
 
 

\end{abstract}

\section{Introduction}

Reactive synthesis (RS) is used to synthesize a strategy (controller) that is provably-correct with respect to a given \ac{ltl} specification. Pneuli and Rosner \cite{Pnueli1989} showed that such an interaction between a controlled agent, called the robot, and its dynamic and uncontrollable environment can be represented as a two-player turn-based zero sum game. Consequently, finding a correct strategy satisfying the specification is equivalent to finding a winning strategy for the robot in the corresponding zero-sum game. In recent years, RS has found applications in several areas such as autonomous vehicles \cite{Hadas2011,Wongpiromsarn2013}, aircraft mission planning \cite{Humphrey2014}, defense \cite{Wu2017} etc.

However, the strategies computed using RS are known to be conservative \cite{bloem2014handle,hagihara2016simple}. This conservativeness may be attributed to the \emph{zero-sum} assumption used to model the interaction. Implicitly, this assumption implies that the environment knows the exact specification of the robot and plays a perfect counter-strategy. However, in many of the applications of RS such as autonomous vehicles, the environment, consisting of other vehicles, may not be adversarial. On the other hand, in defense applications where the environment is adversarial, the enemy may not have complete information regarding the task of the robot. 


\begin{figure}
    \centering
    \includegraphics[scale=0.35]{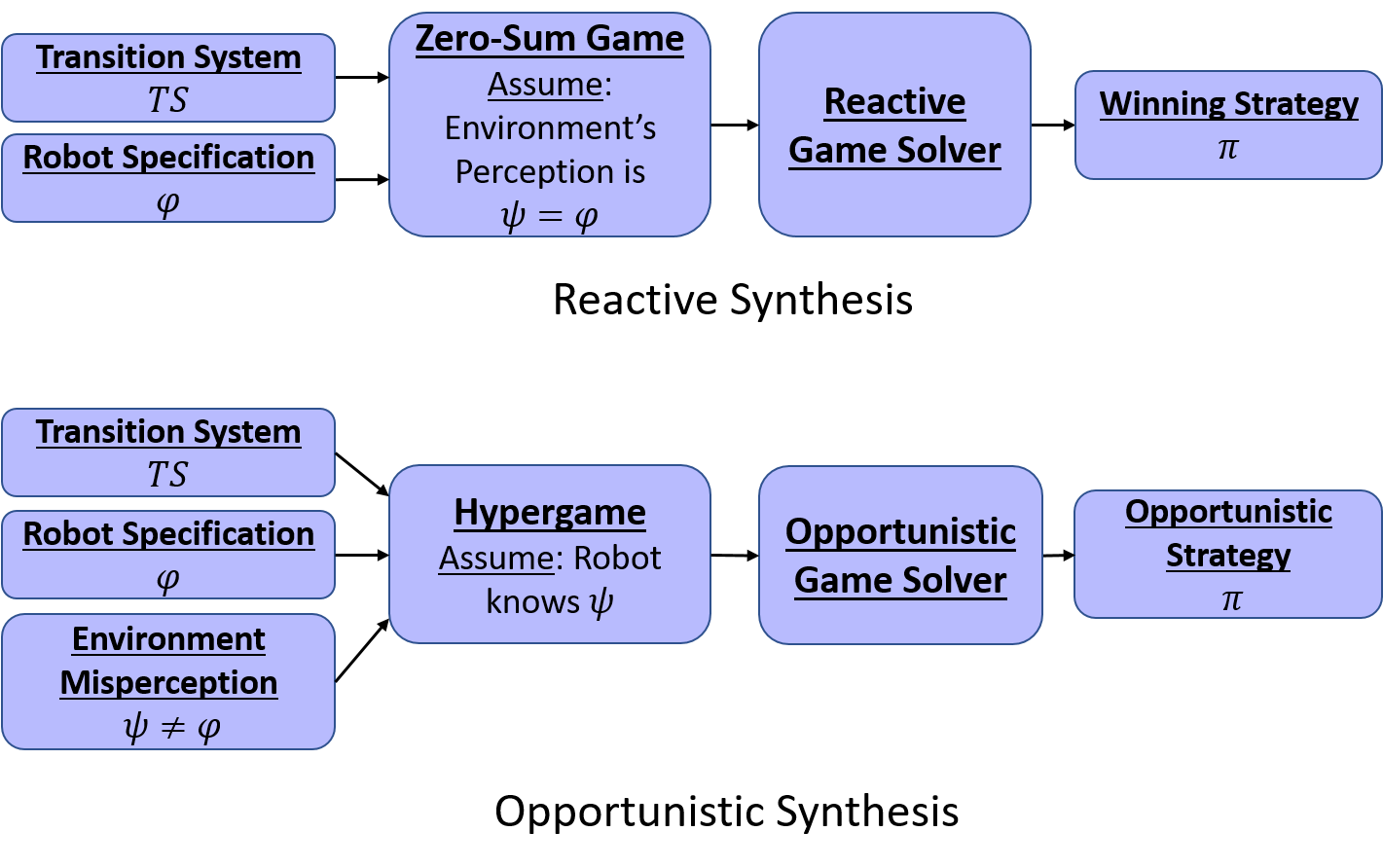}
    \caption{Comparison between Reactive Synthesis and (Proposed) Opportunistic Synthesis. The task of the robot is the \ac{ltl} specification $\varphi$. The environment misperceives the task of robot as the \ac{ltl} specification $\psi \neq \varphi$. The (proposed) hypergame formulation assumes that robot is aware of the information asymmetry.} 
    \label{fig:my_label}
\end{figure}

In this paper, we propose  a method called \emph{opportunistic synthesis} to address the conservativeness of RS and leverage  the information asymmetry between the robot and its adversarial environment, i.e. when one of the two players in the game has more or better information than the other \cite{Kim2004}. Assuming that the environment \emph{does not know} the complete task specification of the robot, we are interested in addressing the following question -- \textit{If the robot is aware of the information asymmetry, then how can it capitalize on the adversary's imperfect counter-strategy to enhance its winning strategy?} The key contributions of this paper are 
\begin{itemize}
\item
  \textbf{Hypergame Model:} We model the interaction between the robot and its environment as a second-level hypergame \cite{Bennett1977}, in contrast to a zero-sum game.  This model represents the ability of the robot to reason about how the environment will behave given that it has incomplete information. By leveraging the knowledge about environment's behavior, we show that the robot can synthesize an opportunistic strategy that dominates\footnote{A strategy is dominant over another if, regardless of what any other players do, the strategy earns a player a larger payoff than the other strategy.} the one computed using RS given complete and symmetrical information. 
\item
  \textbf{Characterization:} The solution of RS partitions the game state-space into winning and losing regions for the robot \cite{Manna1990}. However, we show that under the assumptions of this paper about information asymmetry, the state-space is partitioned into \textit{five} regions.  By assuming that the environment plays  stochastic strategy in regions where it perceives itself to be winning or losing, we show how to construct an \ac{mdp} to represent the hypergame for synthesizing  opportunistic strategy for the robot. 
\end{itemize}

\subsection{Literature}

To the best of our knowledge, this is the first paper investigating the hypergame model to synthesize provably-correct strategies given temporal logic specifications. Hypergames \cite{Bennett1977} are used to model the interactions where one or more players are playing different games because of their misperception of other players' capabilities and/or objectives. Hypergame theory allows the agent to improve its strategy by reasoning about multiple games being played by different players \cite{Bennett1982}. In the literature, this theory has been applied to model strategic interactions such as military conflicts \cite{Bennett1979,Lyn2006}, information security and cyber-physical systems security \cite{Wu2017}. Similar to the hypergame formulation in this paper, Imamverdiyev \cite{Imamverdiyev2013} models the interaction between an attacker and a defender in an information security game using a second-level hypergame. The model is motivated by the information asymmetry that often exists in such a game. He proposes an algorithm to compute equilibrium in a second-level normal form hypergame model. However, the result do not generalize to games with payoffs in terms of evaluation of temporal logic formulas.
Kovach \cite{KovachPhDThesis} is the first to introduce a framework that integrates the temporal logic and hypergame theory. He formalizes the concepts of trust, mistrust and deception in his thesis. However, his work focuses on defining the mathematical framework and verification rather than the strategy synthesis, which is the focus of this paper. 

The games with information asymmetry are a subset of games with incomplete information, with the assumption that at least one player has the correct information \cite{Scerala2017}. The synthesis for incomplete information has been rigorously defined and proved to be EXPTIME-complete for \ac{ltl} by Kupferman and Vardi \cite{Kupferman2000b}. They show that incomplete information does not affect the complexity of synthesis problem. Niu et. al \cite{Niu2018} study the problem of security of cyber-physical systems. They note that zero-sum games are a good tool for the worst-case analysis, but the games with information asymmetry better represent the strategic interactions between the attacker and defender. They approach the minimum violation synthesis problem under \ac{ltl} specifications by modeling the interaction as a concurrent Stakelberg game, which captures the information asymmetry. We take a different approach than Niu et. al by accounting for the robot's ability to reason about the information asymmetry, while treating the interaction as a turn-based game. 

In AI literature, \textit{opportunistic planning} is used in a different context. Cashmore et. al \cite{Cashmore2018} treat opportunities as optional goals in the game, but with high payoffs. Their approach starts by planning for the \textit{must-satisfy} goals of the game and adjusts the plan when an opportunity to satisfy an optional goal becomes available. However, they do not consider any adversarial interaction between the robot and its environment or the ability of robot to reason about environment's perception of it's goals. On the contrary, the \textit{opportunistic synthesis}, as proposed in this paper, assumes the environment to be adversarial, but with incomplete information about its objectives. It identifies the states from where an opportunity to get higher payoff will be available and maximizes the likelihood to reach one of these states.

\section{Reactive Synthesis}

\underline{Notations}: Let $\Sigma$ be a finite alphabet. A sequence of symbols $w=w_0 w_1 \ldots w_n  $ with $w_i\in \Sigma, i=0,1,\ldots, n$ is called a \emph{finite word} and $\Sigma^\ast$ is the set of finite words that can be generated with alphabet $\Sigma$. We denote $\Sigma^\omega$ the set of $\omega$-regular words obtained by concatenating the elements in $\Sigma$ infinitely many times.  Given a set $X$, let $\distr(X)$ be the set of probability distributions over $X$. The indicator function is defined to be $\indicator_X(y)=1$ if $y\in X$ and $0$ otherwise. 


Let $\robot, \adv$ denote the robot (a controlled agent) and its adversary (an uncontrolled environment agent), respectively. Let the tasks of the robot be specified using a subclass of \ac{ltl} formulas, called syntactically co-safe \ac{ltl} formulas \cite{kupferman2001model}. The syntax of \ac{ltl} formulas are given as follows.
{\definition[\ac{ltl}] Let $\ap$ be a set of atomic propositions, the Linear Temporal Logic (LTL) has the following syntax 
\[ \varphi := \top \mid \bot \mid p \mid \varphi \mid \neg\varphi \mid \varphi_1 \land \varphi_2 \mid \bigcirc \varphi \mid \varphi_1 {\cal U} \varphi_2 \mid \lozenge\varphi \] where $\top,\bot$ are universally true and false, respectively, $p \in \ap$ is an atomic proposition, $\bigcirc, \cal U$ and $\lozenge $ denote the temporal modal operators for \textit{next}, \textit{until} and \textit{eventually}.} 

A co-safe \ac{ltl} formula contains only the temporal operator $\bigcirc$, $\cal U$, and $\lozenge$ and can be written in positive normal form. A co-safe \ac{ltl} formula  can equivalently represented by a \ac{dfa} defined as
{\definition[\ac{dfa}] A Deterministic Finite Automaton is defined as a 5-tuple,
\[ \aut = \left< Q, \Sigma, \delta, q_0, \final \right>, \]
where $Q$ is the set of states, $\ap$ is the set of atomic propositions, $\Sigma= 2^{\ap}$ is the set of input symbols, $\delta: Q\times \Sigma \rightarrow Q$ is a deterministic transition function, $q_0$ is the initial state and $\final \subseteq Q$ is the set of accepting states.} 

A word $w\in \Sigma^\ast$ is accepted if and only if $\delta(q_0,w)\in \final$. 
An infinite word $w\in \Sigma^\omega$ satisfying an \ac{ltl} formula $\varphi$ contains a good prefix $w_0w_1\ldots w_n$ that is accepted in the \ac{dfa} corresponding to $\varphi$. Given a co-safe \ac{ltl} formula, the \ac{dfa} accepting finite good prefixes for $\varphi$ can be obtained using tools such as spot \cite{spot2}.

We will assume that all \ac{dfa}s referred to in this paper are complete, i.e. for every state $q \in Q$ and for every input symbol $a \in \Sigma$ the transition function $\delta(q, a)$ is defined. An incomplete \ac{dfa} can be made complete by introducing a sink state and directing all undefined transitions into the sink state. 

The interaction between the robot and its adversary is captured in a two-player turn-based transition system, 

{\definition[Transition System (TS)] The \ac{tsys} is a 6-tuple \[\tsys = \left< S, \act, T, s_0, \ap, \labelfcn \right>\]  where $S = S_\robot \cup S_\adv$ is the set of states partitioned on the basis of the turn of $\robot$ and $\adv$, $\act = \robotact \cup \advact$ is the set of actions for $\robot$ and $\adv$ respectively. The function $T: S \times \act \rightarrow S$ represents the deterministic transition function. The $\ap$ is a set of atomic propositions and $\labelfcn: S \rightarrow 2^{\ap}$ is the labeling function.}

{\definition[Reactive Game] A reactive game between the robot and its adversary defined by a transition system $\tsys$ and a specification automaton $\aut$ representing the language of $\varphi$ is the product transition system given by,
\[ \game(\varphi) = \left< G, \act, \Delta, g_0, F_\varphi 
\right> \]
where $G = S \times Q$, $g_0 = (s_0, \delta(q_0, L(s_0))$ and $\Delta: G \times \act \rightarrow G$ is the transition function such that given the states $g=(s,q)$ and $g'=(s',q')$,  $\Delta(g, a) =  g'$ if and only if $T(s, a) = s'$ and $\delta(q, \labelfcn(s')) = q'$. The set $F_\varphi = S \times \final$ is a set of accepting states.} 

A run in the game $\game(\varphi)$ is an infinite sequence of states $\rho = g_0g_1\ldots$. Given a run $\rho$, the set of states that occur in the run $\rho$ is denoted by $\mathsf{Occ}(\rho) = \{g \in G \mid  \exists i\ge 0, g_i = g\}$. A run is said to be winning for $\robot$ if it satisfies $\mathsf{Occ}(\rho) \cap F_\varphi \ne \emptyset$. If a run is not wining for $\robot$, it is winning for $\adv$.
A state $g \in G$ is said to be winning if the robot can enforce a win from $g$. Otherwise, $g$ is said to be losing. The exhaustive set of winning states for the robot is called winning region of the robot and is denoted by $\win_\robot$. The winning regions for robot and its adversary are mutually exclusive. The winning region for the robot is computed using the Zielonka attractor algorithm \cite{Zielonka1998} as follows: Given a set of final states $\final_\varphi$, 
\begin{enumerate}
    \item Let $\att_0 = \final_\varphi$.
    \item $\att_{k+1} = \att_{k}\cup \{g \in (S_\robot \times Q) \mid \exists a \in \act.\; \Delta(g, a) \in \att_k) \} \cup \{g \in (S_\adv \times Q) \mid \forall a \in \act. \; \Delta(g, a) \in \att_k \}$
    \item Repeat (2) until $\att_{k+1} = \att_k$. Let $\att_{k}=\att^\ast$. 
    \item Let $\att(\final) = \att^\ast$.
\end{enumerate}

We denote the set $\att(\final)$ as the attractor set. The rank of a state $g$ is the smallest level $k$ at which $g \in \att_k$, denoted as $\mathsf{rank}(g) = k$.  The winning region for the robot in the game $\game(\varphi)$ is the set of states in attractor $\att(\final_\varphi)$.


A specification $\varphi$ is said to be realizable for robot over the transition system $\tsys$ if and only if the winning region for the robot, $\win_\robot$, contains the initial state $g_0 \in G$. Otherwise, the specification is said to be unrealizable. 



Given a game $\game(\varphi)$, a stochastic, memoryless strategy for the robot is a function $\pi: \win_\robot \cap (S_\robot\times Q) \rightarrow \distr(\act_\robot)$. A strategy is said to be \textit{almost-sure-winning} if every run $\rho$, produced as a result of robot using strategy $\pi$ and adversary using any feasible strategy $\sigma$, is a winning run with probability one. Given  a state $g \in  \win_\robot \cap (S_\robot\times Q)$, the \textit{almost-sure-winning} strategy $\pi$ for the robot can be given as follows: For each $g\in \win_\robot\cap (S_\robot\times Q)$, let $\mbox{safe}_\robot(g)= \{ a \in \act_\robot \mid \Delta(g, a) \in \win_\robot \}$. Let $\mbox{progress}_\robot(g) = \{a \mid  \mathsf{rank}(\Delta(g, a)) < \mathsf{rank}(g) \}$. Let $\pi(g,a)>0$ for at least one action in $\mbox{progress}_\robot(g)$ and the support of $\pi(g)$ be a non-empty subset of $\mbox{safe}_\robot(g)$. The proof can be found in \cite{fu2016synthesis} 
The  \textit{almost-sure-winning} strategy $\sigma$ of the adversary is a memoryless maximally permissive strategy \cite{bernet2002permissive} and can be defined as follows. Let $\mbox{safe}_\adv(g) = \{ a \in \act_\adv \mid \Delta(g, a) \in \win_\adv  \}$ and the support of $\sigma(g)$ is a  non-empty subset of $\mbox{safe}_\adv(g)$. By definition, an almost-sure winning stochastic strategy for either $\robot$ or $\adv$ is not unique.

Under the assumption of complete observation and complete information, it follows that there exists no strategy for either the robot or its adversary to reach a winning state from a losing state. 
However, when the adversary has incomplete knowledge about the task specification of the robot, we  have a reactive game with asymmetrical information. In such games, we show that the robot can synthesize \textit{opportunistic} strategies that exploit the information asymmetry to enforce a win in an otherwise unrealizable game; had there been no information asymmetry.


\section{Reactive Game under Information Asymmetry}
\subsection{Hypergame}

Consider an interaction between the robot and its adversary where the robot has the \ac{ltl} specification $\varphi$ while its adversary believes that the robot is trying to satisfy a different specification $\psi \ne \varphi$. 

\begin{assumption}
Let $\varphi_1, \varphi_2$ be two \ac{ltl} formulas such that 
\begin{align*}
\varphi &= \varphi_1 \land \varphi_2, \\
\psi &= \varphi_1.
\end{align*}
\end{assumption}

The above assumption means that the adversary knows partial task specification of the robot. The interaction between two players, where at least one of the player has incorrect perception of the true specification of the opponent, can be represented as a  hypergame \cite{Bennett1977,gharesifard2012evolution}. 





{ \definition[Hypergame] A first-level hypergame between two players is represented as a 2-tuple \[ \mathcal{H}^1 = \left< \game_\robot, \game_\adv \right> = \left<\tsys, \{\varphi, \psi\} \right> \] where $\game_\robot = \game(\varphi)$ is the reactive game from the robot's perspective while $\game_\adv = \game(\psi)$ is the game from the adversary's perspective. The tuple $\left<\tsys, \{\varphi, \psi\} \right>$ is an equivalent representation that highlights that both the games $\game_\robot$ and $\game_\adv$ are defined over same transition system $\tsys$ but each player has a different perception of robot's specification.}

When the robot is aware of the existing misperception, i.e. the adversary's belief about the robot's specification $\psi$, we have a second-level hypergame.  

{\definition[Second-Level Hypergame] The second-level hypergame between two players; the robot and its adversary, where only the robot is aware of the misperceived game is represented as \[{\cal H}^2 =  \left< {\cal H}^1, \game_\adv \right>  \] where the robot computes the strategy by solving the hypergame ${\cal H}^1$ and the adversary computes strategy by solving the reactive game $\game_\adv = \game(\psi)$. }

Given the second-level hypergame as defined above, we are interested in the following question

{\question[] 
Assume that the specification $\varphi$ is unrealizable in the reactive game $\game(\varphi)$ with complete information. Then is it possible that when information asymmetry exists, as captured by the hypergame ${\cal H}^2$, the robot can satisfy the specification $\varphi = \varphi_1 \land \varphi_2$ with a high likelihood? If not, then under what conditions can the robot satisfy at least a part of specification, $\varphi_1$ (common knowledge) or $\varphi_2$ (only known to robot)? 

}



To answer the above question, we first define a transition system that captures the information asymmetry compactly and  facilitates game-theoretic analysis and strategic planning.

{\definition[Hypergame Transition System] \label{defn:tsys} Let $\aut_1 = \left< Q_1, \Sigma, \delta_1, q_{10}, F_1 \right>$ and $\aut_2 = \left< Q_2, \Sigma, \delta_2, q_{20}, F_2 \right>$ be the specification automata for \ac{ltl} formulas $\varphi_1, \varphi_2$. Then, the hypergame transition system in its explicit form is a  5-tuple,
\begin{align*}
    {\cal H} = \left< H, \act, \Delta, h_0, {\cal F} \right>
\end{align*}
where $H = S \times Q_1 \times Q_2$ is the set of states and $h_0 \in H$ is the initial state. Given a state $h = (s, q_1, q_2)$ and action $a \in \act$, the transition function $\Delta: H \times \act \rightarrow H$ is given by $\Delta(h, a) = h' = (s', q_1', q_2')$ where $s' = T(s, a), ~q_1' = \delta_1(q_1, \labelfcn(s'))$ and $q_2' = \delta_2(q_2, \labelfcn(s'))$. The set ${\cal F} = (S \times F_1 \times Q_2) \cup (S \times Q_1 \times F_2)$ is the set of final states. 

}

The choice of $\cal F$ as the accepting state set is motivated by the fact that it contains the final state sets ${\cal F}_1 =S\times \final_1\times Q_2, {\cal F}_2 = S\times Q_1\times \final_2 \text{ and } {\cal F}_{12} = S\times \final_1\times \final_2$ of the games $\game(\varphi_1), ~\game(\varphi_2) \text{ and } \game(\varphi_1 \land \varphi_2)$. This facilitates the computation and analysis of winning regions of the game $\game(\varphi)$ and the sub-games $\game(\varphi_1), \game(\varphi_2)$ over the same transition system.  

The outcome of hypergame in this transition system is the run $\rho =  h_0 h_1 h_2 \ldots$. By construction, the run is winning for robot over specification $\varphi$ (resp. $\varphi_1$, $\varphi_2$) if and only if $\mathsf{Occ}(\rho) \cap {\cal F}_{12} \ne \emptyset$ (resp. $\mathsf{Occ}(\rho) \cap {\cal F}_1 \ne \emptyset$, $\mathsf{Occ}(\rho) \cap {\cal F}_2 \ne \emptyset$).  

\subsection{The Partition of States} \label{sec:partition}
Given the hypergame transition system $\cal H$, we are interested in identifying the states from where the robot has a strategy to satisfy $\varphi_1, \varphi_2$ and/or $\varphi$. Therefore, we compute the three winning regions for robot in the reactive games over specifications $\varphi_1, \varphi_2$ and $\varphi$. 
\begin{enumerate}
    \item $\win_\robot(\varphi_1) = \att(\mathcal{F}_1)$ is the set of winning states in the game $\game(\varphi_1)$. 
    \item $\win_\robot(\varphi_2) = \att(  \mathcal{F}_2)$ is the set of winning states in the game $\game(\varphi_2)$. 
    \item $\win_\robot(\varphi) = \att(\mathcal{F}_{12})$ is the set of winning states in the game $\game(\varphi)$. 
\end{enumerate}

We have the following relations between winning regions in three games.

\begin{lemma}
\label{lm1}
Given $\varphi=\varphi_1\land \varphi_2$, it holds that 
 $\win_\robot(\varphi) \subseteq \win_\robot(\varphi_i)$, for $i=1,2$.
\end{lemma}
\begin{proof}
Let $\pi_\varphi$ be the winning strategy for the robot with respect to task $\varphi$. For any state $h\in \win_\robot(\varphi)$, the robot, by exercising $\pi_\varphi$, can enforce a run to visit $\calF_{12}$. Because $\calF_{12}\subseteq \calF_i$, for $i=1,2$, then this run satisfies $\varphi_i$, for $i=1,2$. By definition of winning region, it holds that $h\in \win_\robot(\varphi_i)$, $i=1,2$, witnessed by strategy $\pi_\varphi$. 
\end{proof}

On the contrary, $\win_\adv(\varphi_i) \supseteq \win_\adv(\varphi)$. Thus, if the adversary can ensure to win game  $\game(\varphi_1)$, then even though it does not know $\varphi_2$, it can prevent the robot from satisfying any specification $\varphi_1\land \phi$, where $\phi$ is an arbitrary \ac{ltl} formula.

Next, we define a win-labeling function ${\cal{W}}: H \rightarrow \{W, L\}^3$ that labels each state $h \in H$ with an ordered 3-tuple denoting whether the state $h$ is winning (W) or losing (L) for the robot in the games  $\game(\varphi_1)$, $\game(\varphi_2)$, and $\game(\varphi)$. For example, if a state $h$ is winning for the robot in the game $\game(\varphi_1)$ and the game $\game(\varphi_2)$, but losing in the game $\game(\varphi)$, then its win-label is ${\cal W}(h) = \{W, W, L\}$  \footnote{If $\varphi = \varphi_1 \land \varphi_2$ then for a state $h \in H$ to be winning in the game over $\varphi$, it must be winning over in both the sub-games over $\varphi_1$ and $\varphi_2$ \cite[Lma 1]{kulkarni2018compositional}.}. 

Note that the win-labeling function can assign to every state $h \in H$, a unique label from $2^3 = 8$ possible labels. We analyze each possible label separately to understand which of the objectives $\varphi_1, \varphi_2$ or $\varphi$ should the robot try to satisfy. 

\paragraph{Case I: ${\cal W}(h) = (L, L, L)$} The state $h$ is losing for robot in the games $\game(\varphi_1), \game(\varphi_2)$ and $\game(\varphi)$. That is, the adversary has a winning strategy $\sigma$ that will ensure that the robot can never satisfy $\varphi_1$. Therefore, in this case, the robot can try to satisfy only $\varphi_2$, but will never be able to satisfy $\varphi$. 

\paragraph{Case II: ${\cal W}(h) = (L, W, L)$} The state is losing for robot in the games $\game(\varphi_1)$ and $\game(\varphi)$, but winning in game over $\varphi_2$. That is, the adversary has a winning strategy $\sigma_W$ that will ensure that robot can never satisfy $\varphi_1$. Therefore, in this case, the robot must satisfy only $\varphi_2$, and not $\varphi$.

\paragraph{Case III: ${\cal W}(h) = (W, L, L)$} The state is losing for robot in the games $\game(\varphi_2)$ and $\game(\varphi)$, but winning in game over $\varphi_1$. That is, the adversary believes that it has lost the game and can be assumed to play a random strategy, $\sigma_L$. In this case, the robot may try to satisfy $\varphi$, while staying within the winning region of game $\game(\varphi_1)$. 

\paragraph{Case IV: ${\cal W}(h) = (W, W, L)$} The state is winning for robot in the games $\game(\varphi_1)$ and $\game(\varphi_2)$, but losing in game $\game(\varphi)$. That is, the adversary believes that it has lost the game. This state presents an interesting decision problem where robot must decide whether to \textit{try} satisfying $\varphi$ or satisfy just one of the specifications, $\varphi_1$ or $\varphi_2$. 

\paragraph{Case V: ${\cal W}(h) = (W, W, W)$} This is a trivial case, which is the conventional reactive game. The robot can exercise the winning strategy for $\game(\varphi)$ regardless of the strategy and perception of the adversary.

\paragraph{Cases VI-VIII: ${\cal W}(h) = (L, L, W)$, $(L, W, W)$, or $(W, L, W)$} These cases are not possible, because the robot must be winning in $\varphi_1$ and $\varphi_2$ to be winning in $\varphi$ (See Lemma~\ref{lm1} and \cite{kulkarni2018compositional}).

\subsection{Synthesizing opportunistic and reactive strategies}

Recall that the winning regions $\win_\robot(\cdot)$ we computed in previous subsection ensures a win in the respective reactive games. The winning strategies based on these winning regions do not exploit the information asymmetry. Hence, to identify the opportunities generated due to the information asymmetry,  we  make certain assumption about the strategy of the adversary. The assumptions, when the adversary \emph{believes} that the current state is losing for itself, i.e. the win-label of the state is of the form ${\cal W}(h) = (W, \cdot, \cdot)$ where $\cdot$ means it can either be $L$ or $W$, are as follows. 

\begin{assumption}
For a state $h \in H$ with the win-label $ {\cal W}(h) = (W, \cdot, \cdot)$ the adversary plays a stochastic strategy $\sigma_L: H \rightarrow \distr(\act_\adv)$.  
\end{assumption}


\begin{assumption}
For a state $h \in H$ with the win-label ${\cal W}(h) = (L, \cdot, \cdot)$, the adversary plays an almost-sure winning strategy $\sigma_W: H \rightarrow \distr(\act_\adv)$ in the game $\game(\varphi_1)$. 
\end{assumption}

{
\remark In this work, we assume that the adversary's losing and winning strategies, $\sigma_L, \sigma_W$, are known to the robot. This assumption may be relaxed if the robot can learn the strategy using model-based reinforcement learning \cite{brafman2002r} or strategy inference  \cite{Paschalidis2018}. This extension will be considered in our future work. }

To develop opportunistic strategy for the robot, we assume that the robot receives a payoff $r_1 \in \mathbb{R}_{> 0}$ if it satisfies $\varphi_1$ and payoff $r_2 \in \mathbb{R}_{> 0}$ if it satisfies $\varphi_2$. Its payoff for satisfying $\varphi$ is $r \in \mathbb{R}_{> 0}$ with the constraint $r \geq r_1 + r_2$ for the decision problem to make sense. The adversary receives the payoff $-r_1$ if the robot satisfies $\varphi_1$ and payoff $r_1$ otherwise. Note that when $r_1 = r_2$, then the robot is considered to be indifferent to satisfying either $\varphi_1$ or $\varphi_2$. Otherwise, if $r_1 > r_2$, then $\varphi_1$ is strictly preferred over $ \varphi_2$, and vice versa.


\begin{definition}[Hypergame for Opportunistic Synthesis] \label{defn:mdp}
Given the hypergame transition system ${\cal H} =  \left< H, \act, \Delta, h_0, {\cal F} \right>$ and the strategy $\sigma = (\sigma_W, \sigma_L$) used by the adversary given its perceived winning and losing states, the opportunistic planning reduces to an \ac{mdp}, defined by 
\[
{\cal H}^{\sigma} = \langle H_\robot, \act_\robot \cup \{\stopgame\},  P, h_0, R\rangle,
\] 
where $ H_\robot=S_\robot \times Q_1\times Q_2$ is a set of states where the robot makes a move, and the probabilistic transition function $P$ and the payoff function $R$ are defined based on the win-label of a state $h \in H_\robot$ as follows,
\begin{itemize}
\item ${\cal{W}}(h) \in (L, L, L)$:
    \begin{itemize}
        \item Enabled Actions: All feasible actions are enabled. The special action $\stopgame$ is not enabled.
        
        \item Transition probability function is given by
            \[
                P(h'\mid h, a_\robot)= \sum_{a_\adv \in \act_\adv} \indicator_{\{h'\}}(\Delta( h, (a_\robot, a_\adv)))\sigma_W(h , a_\adv)
            \]
    \end{itemize}
    
\item  ${\cal{W}}(h) \in (W, L, L) $:
    \begin{itemize}
        \item Enabled Actions: Only actions that have \textit{zero} probability of reaching a state with win-label $(L, L, L)$ are enabled. In other words, the robot is forced to stay within the winning region for $\win_\robot(\calF_1)$ and at least satisfy $\varphi_1$. The special action $\stopgame$ is enabled.
        
        \item Transition probability function is given by
            \[
                P(h '\mid h , a_\robot)= \sum_{a_\adv\in\act_\adv} \indicator_{\{h '\}}(\Delta(h , (a_\robot,a_\adv)))\sigma_L(h , a_\adv)
            \]
            For action $\stopgame$, the game transitions to a sink state $\sink_1$ with probability one, $P(\sink_1\mid h, \stopgame)=1$.
            
            \item Payoff function:  The payoff for reaching the sink state $\sink_1$ is defined as $R(\sink_1) = r_1$.
    \end{itemize}
    
\item ${\cal{W}}(h) \in (L, W, L)$: 
    \begin{itemize}
        \item All states are labeled as absorbing, i.e. $P(h' \mid h, a_\robot) = 0$. This is because the adversary will play its winning strategy $\sigma_W$ in the game $\game(\varphi_1)$ and the robot will never satisfy $\varphi_1$. 
        \item Payoff function: The payoff for reaching the state $h$ is defined as  $R(h) = r_2$. 
    \end{itemize}


\item ${\cal{W}}(h) \in (W, W, W)$: 
    \begin{itemize}
        \item In this partition, the robot must switch to its winning strategy in the game $\game(\varphi)$. 
        \item Payoff function: The payoff for reaching the state $h$ is defined as  $R(h) = r$. 
    \end{itemize}

\item $\mathcal{W}(h ) \in (W,W,L)$: 
    \begin{itemize}
        \item Enabled actions: Any action that does not lead into partition $(L, L, L)$ is enabled. The special action $\stopgame$ is also enabled. 
        \item For the action $\stopgame$, the robot transitions to a sink state $\sink$. The payoff for reaching the sink state $\sink$ is defined as $R(\sink) =\max( r_1,r_2)$.
        \item Transition probability function is 
        \[
            P(h '\mid h , a_\robot)= \sum_{a_\adv\in\act_\adv} \indicator_{\{h '\}}(\Delta(h , (a_\robot,a_\adv)))\sigma_L(h , a_\adv)
        \]
        
    \end{itemize}
\end{itemize}
\end{definition}

The optimal opportunistic strategy $\pi$ for the robot is the one that solves 
\[
\max_{\pi}\mathbb{E}\left[ 
\sum_{t=1}^T R(h _t)
\right]
\]
where $T$ is the first time when a sink state is reached. The rationale behind defining sink states is to provide the robot with a mechanism to decide whether it wants to explore the state space to find an opportunity or settle for a sub-optimal payoff by satisfying a sub-specification. We define the set of states $\{h \mid {\cal{W}}(h)\in \{ (L,W, L) ,(W,W, W)\}\} \cup \{\sink_1, \sink\}$ as absorbing in the hypergame \ac{mdp}.

\begin{lemma}
By following the optimal strategy in the \ac{mdp} $\mathcal{H}^\sigma$, 
the total payoff  is finite.
\end{lemma}

\begin{proof}
We prove this case by case: When the initial state $h_0$ is labeled $(L,L,L)$, then by following the optimal strategy, the robot can reach a state with labels in $\{(W, W, L), (W,L,L), (L,W,L), (W,W,W)\}$ or stay in $(L,L,L)$. In the case of staying in $(L,L,L)$, the robot receives a payoff of zero. The total payoff is bounded. 
When it reaches a state with labels in $\{(W, W, L), (W,L,L), (L,W,L), (W,W,W)\}$, the robot is ensured to at least win one of the games as its feasible actions are restricted to ensure staying in the winning regions it is currently in. Then, in cases when the label is in $\{(W, W, L), (W,L,L)\}$, it will either settle down to win one of the games by taking the action $\stopgame$ or reach a state with label in $\{(L,W,L), (W,W,W)\}$, which are absorbing and have finite payoffs. Thus, the total payoff is bounded and the planning to maximize the  total payoff without discounting is well-defined.
\end{proof}

\section{Case Study}

\begin{figure}
    \centering
    \includegraphics[scale=1]{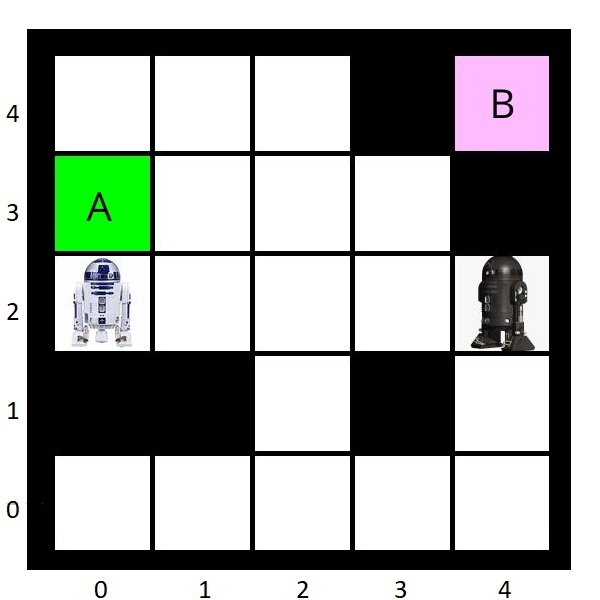}
    \caption{A $5 \times 5$ gridworld with R2D2 at the cell $(0, 2)$ and IDroid at $(4, 2)$. The objective of R2D2 to visit the cell $(0, 3)$ labeled A is known to both players. The objective of R2D2 to visit the cell $(4, 4)$ labeled B is not known to IDroid. The black cells represent the obstacles.}
    \label{fig:gridworld}
\end{figure}

We illustrate our approach using a gridworld example as shown in \fig{fig:gridworld}, with two robots - R2D2 and Imperial Droid (IDroid). R2D2 is the controllable robot, whereas IDroid is adversarial. The objective of R2D2 is to visit two regions, $A$ (green) and $B$ (blue), while avoiding obstacles $O$ (black), whereas the objective of IDroid is to prevent R2D2 from completing its task. We consider the case where IDroid \textit{misperceives} that R2D2's task is to visit only region $A$. Therefore, using the \ac{ltl} notation, the specification of R2D2 is $\varphi = (\neg O ~\mathcal{U}~A) \land (\neg O~\mathcal{U}~B)$, whereas the misperception of IDroid about R2D2's task is $\psi = \neg O~\mathcal{U}~A$. This defines the information asymmetry in the interaction. Furthermore, we restrict the actions of R2D2 and IDroid for illustration purposes as follows, (we will use the symbol $R$ to denote R2D2 and $E$ to denote IDroid to maintain consistency with the notation of the paper)
\begin{align*}
    \act_\robot &= \{\text{N, S, E, W, NE, NW, SE, SW}\} \\
    \act_\adv &= \{\text{N, S, E, W, STAY}\}
\end{align*}

Given 20 obstacle-free cells of gridworld in \fig{fig:gridworld} and the action-set, we construct the transition system (\defn{defn:mdp}) with $20 \times 20 \times 2 = 800$ states. The automaton equivalent to $\neg O ~\mathcal{U}~X$ for $X = A, B$ is shown in the \fig{fig:automaton}.  We prune unsafe actions that drive the robot to the obstacle and thus exclude the transitions labeled $O$ and the state $2$ in computing the transition system. Therefore, the hypergame transition system has $800 \times 2 \times 2 = 3200$ states, where we keep track of both sub-specification using two automata.  Consequently, each sub-game, $\game(\varphi_1)$ and $\game(\varphi_2)$, has $800 \times 2 \times 1 = 1600$ final states and the game $\game(\varphi)$ has $800$ final states. The attractor computation for each of the three games generates the winning regions with sizes: $|\win_\robot(\varphi_1)| = 2491$, $|\win_\robot(\varphi_2)| = 2527$, and $|\win_\robot(\varphi)| = 1831$.

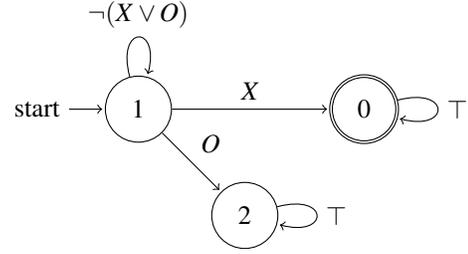
\begin{figure}
    \centering
\begin{tikzpicture}[shorten >=1pt,node distance=2cm,on grid,auto,scale =1, transform shape] 
   \node[state, initial] (1)   {$1$}; 
      \node[state] (2) [below right =of 1] {$2$};
   \node[state, accepting, node distance = 3cm] (0) [ right =of 1] {$0$};
\path[->] 
    (1) edge node{$X$} (0)
 (1) edge [loop above] node{$\neg (X \lor O)$} (1)
     (1) edge node{$O$} (2)
      (2) edge [loop right] node{$\top$} (2)
            (0) edge [loop right] node{$\top$} (0);
\end{tikzpicture}
\caption{The automaton for $\neg O~\mathcal{U}~X$, where $X\in \{A,B\}$.}
    \label{fig:automaton}
\end{figure}

Given the three winning regions, we first validate that the state-space is indeed paritioned in five regions as discussed in \sect{sec:partition}. For every state in the hypergame transition system, we assign a win-label to it by determining the winning regions in which the state appears. The result is tabulated in \ref{table:partitions}. We observe that the state-space is partitioned into exactly five regions. 

\begin{table}[]
\centering
\begin{tabular}{|c|c|}
\hline
\textbf{Partition} & \textbf{Number of States} \\ \hline
(W, W, W)          & 1831                      \\ \hline
(W, W, L)          & 181                       \\ \hline
(W, L, L)          & 479                       \\ \hline
(L, W, L)          & 515                       \\ \hline
(L, L, L)          & 194                       \\ \hline
(W, L, W)          & 0                         \\ \hline
(L, W, W)          & 0                         \\ \hline
(L, L, W)          & 0                         \\ \hline
\end{tabular}
\caption{Partition of game state-space due to information asymmetry.}
\label{table:partitions}
\end{table}

{\remark It is not necessary that the states will always be partitioned into \textit{exactly} five regions. For instance, consider an adversary with only ``STAY" action, then the state-space will be partitioned into exactly 2 regions.}

Using the five partitions, we construct the hypergame \ac{mdp} as defined in \defn{defn:mdp}. As expected, the \ac{mdp} has $1600$ states, where R2D2 makes a decision. We define the stochastic strategies for IDroid as follows: for every state with win-label of $(L, \cdot, \cdot)$, we assume $\sigma_W$ to be a uniform distribution over all safe actions; i.e. the actions that lead to another state with a win-label of type $(L, \cdot, \cdot)$, with probability \textit{one}. We define $\sigma_L$ by assigning a random distribution over all feasible actions from a state within partitions $(W, \cdot, \cdot)$. Given the hypergame \ac{mdp} states and the adversary strategy $\sigma$, the transition probabilities are determined based on win-label of the state and the corresponding expression for $P(h' \mid h, a)$ provided in \defn{defn:mdp}. We compute the value function and opportunistic strategy using the standard value iteration algorithm \cite{Puterman1994}. 



Next, we illustrate the decision process in the hypergame \ac{mdp}.
Let the initial configuration be such that R2D2 is at the cell $(0, 2)$, and IDroid is at $(4, 2)$ as shown in \fig{fig:gridworld}. Therefore, the initial state in the hypergame \ac{mdp} is $h_0 = (((0, 2), (4, 2), 0), 1, 1)$. We define the payoff for reaching goal $A$ as $r_1 = 200$ and that for reaching goal $B$ as $r_2 = 100$. With this initial configuration we simulate the interaction between R2D2 and IDroid, where R2D2 uses the opportunistic strategy $\pi$ and IDroid uses the strategy $\sigma$. We run the simulation for 100 times. We will use one of the runs obtained from simulation, as given below
\begin{enumerate}
    \item State: $(((0, 2), (4, 2), 0), 1, 1)$, win-label: $(W, L, L)$
    \item State: $(((0, 3), (3, 2), 0), 0, 1)$, win-label: $(W, L, L)$
    \item State: $(((1, 2), (2, 2), 0), 0, 1)$, win-label: $(W, W, W)$
\end{enumerate}

\begin{table}[]
\centering
\vspace{1em}
\begin{tabular}{|c|c|c|c|c|}
\hline
\textbf{Act}                & \textbf{Next State}                  & \textbf{Partition} & \textbf{Prob} & \textbf{Value}  \\ \hline
\multirow{3}{*}{\textbf{N}} & \textbf{(((0, 3), (4, 2), 0), 0, 1)} & \textbf{(W, L, L)} & \textbf{0.03} & \textbf{288.99} \\ \cline{2-5} 
                            & \textbf{(((0, 3), (3, 2), 0), 0, 1)} & \textbf{(W, L, L)} & \textbf{0.36} & \textbf{290.20} \\ \cline{2-5} 
                            & \textbf{(((0, 3), (4, 1), 0), 0, 1)} & \textbf{(W, W, W)} & \textbf{0.61} & \textbf{288.99} \\ \hline
\multirow{3}{*}{E}          & (((1, 2), (4, 1), 0), 0, 1)          & (W, W, L)          & 0.25          & 0               \\ \cline{2-5} 
                            & (((1, 2), (3, 2), 0), 1, 1)          & (W, L, L)          & 0.73          & 297.41          \\ \cline{2-5} 
                            & (((1, 2), (4, 2), 0), 1, 1)          & (W, W, L)          & 0.02          & 0               \\ \hline
\multirow{3}{*}{NE}         & (((1, 3), (3, 2), 0), 1, 1)          & (W, L, L)          & 0.38          & 259.42          \\ \cline{2-5} 
                            & (((1, 3), (4, 2), 0), 1, 1)          & (W, W, L)          & 0.18          & 285.03          \\ \cline{2-5} 
                            & (((1, 3), (4, 1), 0), 1, 1)          & (W, W, L)          & 0.44          & 299.25          \\ \hline
\end{tabular}
\caption{A decision table for state $(((0,2),(4,2),0),1,1)$ with value $285.03$ and strategy to choose action ``N".}
\label{table:transition}
\end{table}

To get some insight into the decision process, observe the Table~\ref{table:transition}, which shows the enabled actions, possible next states and their respective partitions, the probability of reaching those states and the value of those states. Based on the value iteration, the value of initial state is $285.03$, while the optimal strategy is to select action ``N", which has a high likelihood to reach a $(W, W, W)$ state. Note that by choosing action ``E", if the robot reaches a state with value $0$, then it chooses to settle for sub-optimal payoff of $r_1 = 200$ by satisfying only $\varphi_1$. Hence, the action ``N" is preferred over ``E". A similar argument can be given for the action ``NE". 

We now point out the key advantage of the opportunistic synthesis over reactive synthesis. Observe that the initial state is losing in the game $\game(\varphi)$ for R2D2. Therefore, if R2D2 uses reactive synthesis approach, it will give up instantaneously and get no payoff. On the contrary, with opportunistic synthesis, R2D2 could leverage the misperception of IDroid to start from a losing state in $\game(\varphi)$ and reach a winning state in the game. 

We also highlight that the construction of hypergame \ac{mdp} is such that R2D2 behaves rationally and tries to maximize the payoff before settling down with a sub-optimal payoff. Given the initial state in partition $(W, L, L)$, it could have chosen the $\stopgame$ action and switched to the winning strategy in $\game(\varphi_1)$ to get a payoff of $r_1 = 200$. Instead, it preferred to explore for an opportunity to get the payoff of $r = r_1 + r_2 = 300$. 


We conclude this section by counting the number of states with opportunities. This is done by counting the number of \ac{mdp} states with non-zero value. Recall that we label the sink states in the \ac{mdp} as absorbing with a fixed payoff. Therefore, they always have fixed value of \textit{one}. We find that there are a total of $1245$ absorbing states and $312$ states with opportunities. This implies that there are $43$ states with no opportunities. In other words, not all losing states in the reactive game $\game(\varphi)$  have opportunities.

\section{Discussion and Conclusion}

In this paper, we have introduced a novel strategy synthesis approach---\textit{opportunistic synthesis}---to solve reactive games under information asymmetry. By modeling the misperception in the interaction between the robot and its adversarial environment as a hypergame and the corresponding decision problem as a hypergame \ac{mdp}, we identify opportunities by maximizing the expected value to reach a winning state in the reactive game from a losing state. This primarily results in a larger number of states from which the robot can satisfy its specification, $\varphi$. 
As a bonus, the approach also allows us to compute the states from which the robot has an opportunity to satisfy a partial specification, $\varphi_1$ or $\varphi_2$.

From a computational point of view, the opportunistic synthesis has the same (time and space) complexity as that of the reactive synthesis. The number of computations in our approach is related to that of  reactive synthesis by only a constant scaling factor; because we require three attractor computations and a value iteration instead of a single attractor computation. Note that the definition of hypergame transition system in \defn{defn:tsys} dispenses with constructing different game structures for computing winning regions of the sub-games. 

We have presented the preliminary results of our investigation, in what we believe to be the first work, in the use of 
hypergame model to study a reactive game with information asymmetry. To restrict the hypergame model to second-level, we have introduced two assumptions that, \textit{the adversary has partial information regarding robot's specification} and \textit{the robot knows adversary's losing strategy $\sigma_L$ and winning strategy $\sigma_W$}. The relaxation of these assumptions opens up two directions for future research. The first one investigates  opportunistic synthesis when adversary perceives the robot's objective as $\psi \neq \varphi$, i.e. the language of $\psi$ may be a subset, superset or disjoint with the language of $\varphi$. The second direction investigates the use of policy inference techniques for the robot to learn the adversary's strategies $\sigma_L$ and $\sigma_W$. 

\bibliographystyle{ieeetran}
\bibliography{CDC2019Submission}
\end{document}